\documentclass[12pt]{article}  

\usepackage{graphicx}
\usepackage{amsmath}
\usepackage{amsfonts}
\usepackage{amssymb}
\usepackage{amsthm}
\usepackage{color}
\usepackage{centernot}
\usepackage{mathtools}
\usepackage{stmaryrd}

\newtheorem{theorem}{Theorem}
\newtheorem{proposition}{Proposition}
\newtheorem{corollary}{Corollary}
\newtheorem{lemma}{Lemma}

\begin{document}

\title{Spontaneously stochastic Arnold's cat}
\date{\today}
\author{Alexei A. Mailybaev\footnote{Instituto de Matem\'atica Pura e Aplicada -- IMPA, Rio de Janeiro, Brazil. E-mail: alexei@impa.br} \and Artem Raibekas\footnote{Instituto de Matem\'atica e Estat\'istica, UFF, Niter\'oi, Brazil. E-mail: artemr@id.uff.br}}

\maketitle

\begin{abstract}
We propose a simple model for the phenomenon of Eulerian spontaneous stochasticity in turbulence. This model is solved rigorously, proving that infinitesimal small-scale noise in otherwise a deterministic multi-scale system yields a large-scale stochastic process with Markovian properties. Our model shares intriguing properties with open problems of modern mathematical theory of turbulence, like non-uniqueness of the inviscid limit, existence of wild weak solutions and explosive effect of random perturbations. Thereby, it proposes rigorous, often counterintuitive answers to these questions. Besides its theoretical value, our model opens new ways for the experimental verification of spontaneous stochasticity, and suggests new applications beyond fluid dynamics.
\end{abstract}

\section{Introduction}
Scaling symmetries of space and time shape the modern theory of developed turbulence~\cite{frisch1999turbulence}, which assumes that equations of motion for a velocity field $\mathbf{u}(\mathbf{r},t)$ are invariant with respect to the scaling transformations
	\begin{equation}
	\label{eq1}
	t,\mathbf{r},\mathbf{u} \ \mapsto\  \lambda^{1-h} t,\lambda\mathbf{r},\lambda^h\mathbf{u}
	\end{equation}
for arbitrary $\lambda > 0$ and $h \in \mathbb{R}$. Notice that this property refers to a wide (so-called inertial) interval of scales, at which both the forcing and viscous terms are negligible. Multi-scale systems of this kind may possess a fascinating property of \textit{spontaneous stochasticity}: a small-scale initial uncertainty develops into a randomly chosen large-scale state in a finite time, and this behavior is not sensitive to the nature and magnitude of uncertainty~\cite{lorenz1969predictability,leith1972predictability,ruelle1979microscopic,eyink1996turbulence,falkovich2001particles,boffetta2001predictability,palmer2014real,thalabard2020butterfly}. 

A simpler form of this phenomenon is the \textit{Lagrangian spontaneous stochasticity} (LSS) of particle trajectories in a turbulent (non-differentiable) velocity field, also known as the Richardson super-diffusion~\cite{frisch1999turbulence,falkovich2001particles}: two particles diverge to distant random states in finite time independently of their initial separation. Another intriguing form is the \textit{Eulerian spontaneous stochasticity} (ESS) of the velocity field itself: an infinitesimal small-scale noise triggers stochastic evolution of velocity field at finite scales and times. The consequences are both theoretical, revising the role of stochasticity in multi-scale classical systems, and practical, e.g. its implications for weather prediction~\cite{palmer2019stochastic,palmer2000predicting}. 
The ESS suggests a potentially new path for understanding the inviscid limit in the developed (Navier--Stokes) turbulence, which copes with a number of paradoxes like the recently discovered wild and non-unique dissipative weak solutions; see e.g.~\cite{buckmaster2021convex,de2021weak}.  Unlike the LSS, which can been studied in various models~\cite{bernard1998slow,eijnden2000generalized,kupiainen2003nondeterministic,eyink2013flux,drivas2017lagrangian,drivas2020statistical,eyink2020renormalization}, the current knowledge on the ESS is mostly limited to numerical simulations~\cite{palmer2014real,fjordholm2016computation,mailybaev2016spontaneously,biferale2018rayleigh,mailybaev2017toward,thalabard2020butterfly}. A rigorous theory of ESS remains elusive due to its sophisticated (infinite-dimensional) character. 

In this paper, we propose an artificial model, which is constructed as an infinite-dimensional extension of the (hyperbolic) Arnold's cat map~\cite{arnold1968ergodic} and yields a rigorously solvable example of ESS. This model is a formally deterministic system with a scaling symmetry, which possesses non-unique (uncountably many) solutions, including analogues of wild solutions known for the Euler equations of incompressible ideal fluid~\cite{buckmaster2021convex}. However, solutions are made unique by introducing a viscous-like regularization. By mimicking the Navier--Stokes turbulence~\cite{bardos2013mathematics}, we study the inviscid limit and we prove that it exists for subsequences, but yields uncountably many limiting solutions depending on a chosen subsequence. Then, we prove that adding a random perturbation as a part of the regularization yields a unique inviscid limit in the stochastic sense, i.e., it yields a unique and universal probability measure solving the original formally deterministic system with deterministic initial conditions. This probability measure defines a stochastic process with Markovian properties, and its universality means that it does not depend on a specific form of a random perturbation. The counterintuitive property of this spontaneously stochastic solution is that it assigns equal probability (uniform probability density) to all non-unique solutions. The rigorous answers produced by our model shed light on new ways of understanding the problem of non-uniqueness in the developing mathematical theory of turbulence~\cite{de2021weak}.


The paper has the following structure. Section~\ref{sec2} introduces the model and describes basic properties of non-unique solutions. Section~\ref{sec3} 
defines regularized solutions and studies non-unique inviscid (subsequence) limits. Section~\ref{sec4} introduces random regularization and formulates our main result on the existence and uniqueness of a spontaneously stochastic solution, which is proved in Section~\ref{sec5}. Section~\ref{sec6} investigates the convergence issues and presents results of numerical simulations. Further applications of obtained results are discussed in Section~\ref{sec7}.

\section{Model}\label{sec2}
We consider variables $u_n(t)$ depending on time $t$ and integer indices $n \in \mathbb{Z}^+ = \{0,1,2,\ldots\}$. One can see these variables as describing a multi-scale system with a geometric sequence of spatial scales $\ell_n = \lambda^{-n}$ for some $\lambda > 0$. In this case, the discrete analogue of scaling symmetry (\ref{eq1}) with $h = 0$ becomes
	\begin{equation}
	\label{eq2}
	t,u_n \ \mapsto \ \lambda t,u_{n+1},
	\end{equation}
where the index shift $n \mapsto n+1$ reflects the spatial scaling relation $\ell_n = \lambda\ell_{n+1}$. Notice that (\ref{eq2}) is the symmetry of the Euler equations for incompressible ideal fluid, in which case the variable $u_n$ can be introduced by low/high-pass filters or wavelet transforms of the velocity field in the range of scales between $\ell_n$ and $\ell_{n+1}$~\cite{frisch1999turbulence}. 

We construct an artificial model with symmetry (\ref{eq2}) by setting $\lambda = 2$ and defining variables $u_n(t)$ on the two-dimensional torus $\mathbb{T}^2 = \mathbb{R}^2/\mathbb{Z}^2$ at discrete times  
	\begin{equation}
	\label{eq3}
	t \in \tau_n \mathbb{Z}^+ 
	= \{0,\tau_n,2\tau_n,\ldots \},
	\quad
	\tau_n = 2^{-n},
	\end{equation}
where $\tau_n$ is interpreted as the ``turn-over'' time at scale $\ell_n$. As shown in Fig.~\ref{fig1}, all scales and corresponding times define the self-similar lattice 
	\begin{equation}
	\label{eq3b}
	\mathcal{L} = \{(n,t): n \in \mathbb{Z}^+,\ t \in \tau_n \mathbb{Z}^+\}.
	\end{equation}
Our model is defined by the deterministic relation
	\begin{equation}
	\label{eq4}
	u_n(t+\tau_n) = Au_n(t)+Au_{n+1}(t) \ \mathrm{mod}\ 1,
	\end{equation}
where the symmetric $2 \times 2$ matrix $A$ defines the Arnold's cat map~\cite{arnold1968ergodic}
	\begin{equation}
	\label{eq6}
	A: (x,y) \mapsto (2x+y,x+y) \ \mathrm{mod}\ 1,\quad (x,y) \in \mathbb{T}^2.
	\end{equation}
Relation (\ref{eq4}) defines evolution at scale $\ell_n$ over a single turn-over time $\tau_n$. Here we limited the inter-scale couplings to the same and smaller scales, $\ell_n$ and $\ell_{n+1}$, and took advantage that the map $A$ is a linear, hyperbolic, invertible and area-preserving. These properties greatly facilitate analysis of the model, and we discuss further generalizations later. Relation (\ref{eq4}) is 
invariant with respect to the scaling symmetry (\ref{eq2}). The resulting structure of whole system is presented schematically in Fig.~\ref{fig1}. 

\begin{figure}[t]
\centering
\includegraphics[width=0.6\textwidth]{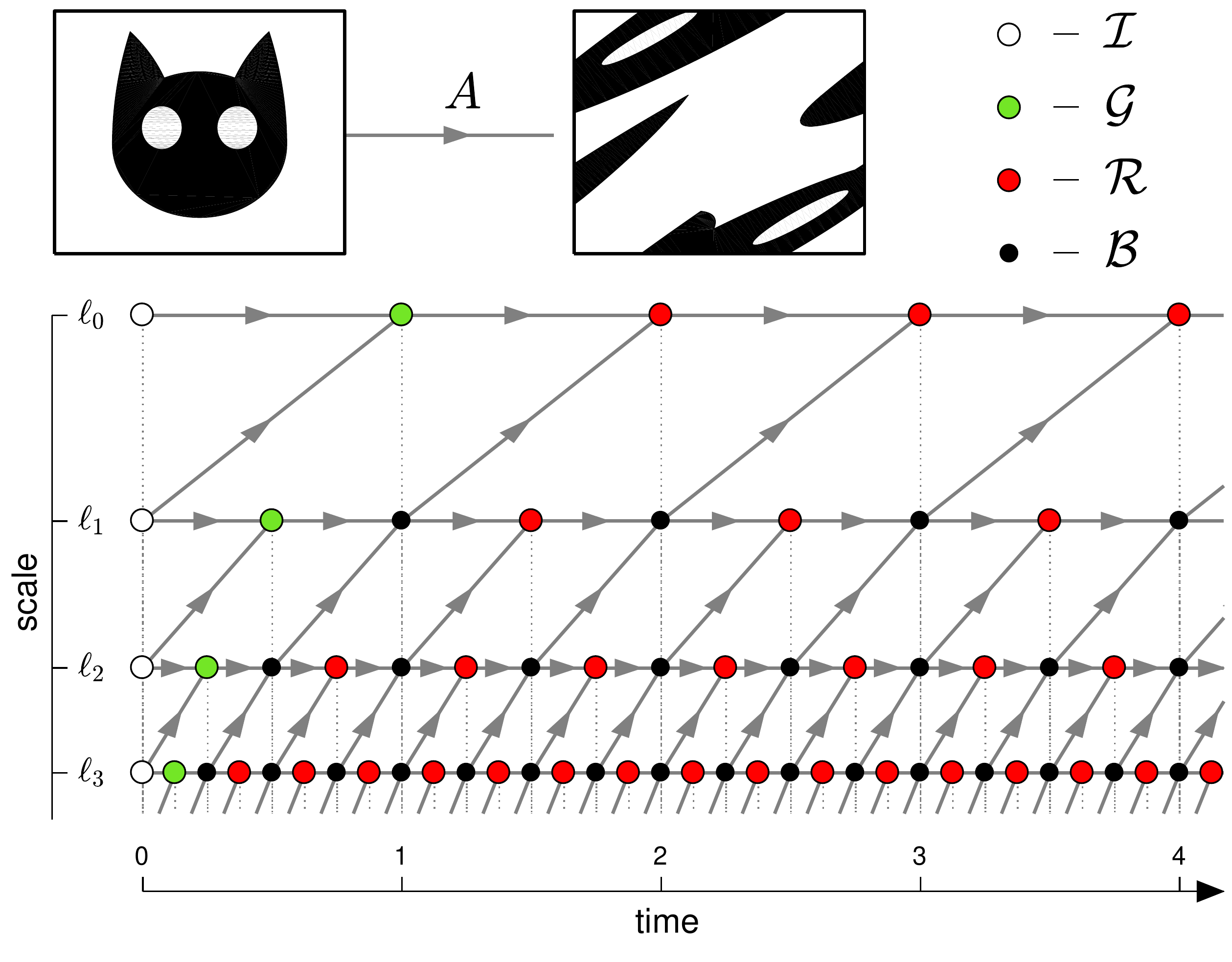}
\caption{Structure of the multi-scale map for the variables $u_n(t)$ corresponding to scales $\ell_n$ and discrete times $t \in \tau_n \mathbb{Z}^+$. Gray arrows represent the Arnold's cat map (shown on the top of the figure), which appear in the coupling relation (\ref{eq4}) and correspond to one turn-over time $\tau_n$. White circles correspond to initial conditions. Green ($\mathcal{G}$), red ($\mathcal{R}$) and small black ($\mathcal{B}$) circles denote, respectively, the next-time variables, the variables taking arbitrary values in Proposition~\ref{theorem1}, and the remaining variables.}
\label{fig1}
\end{figure}

We assume an arbitrary deterministic initial condition 
	\begin{equation}
	\label{eq7}
	u_n(0) = u_n^0,\quad n \in \mathbb{Z}^+.
	\end{equation}
We say that the infinite sequence $\left(u_n(t)\right)_{(n,t) \in \mathcal{L}}$ is a solution of the initial value problem, if it satisfies relations (\ref{eq4})--(\ref{eq7}) for all $(n,t) \in \mathcal{L}$. For describing all solutions, we split the lattice, $\mathcal{L} = \mathcal{I} \cup \mathcal{G} \cup \mathcal{R} \cup \mathcal{B}$, as shown in Fig.~\ref{fig1}. Here $\mathcal{I} = \left\{(n,0): n \in \mathbb{Z}^+\right\}$ are indices of initial conditions and $\mathcal{G} = \left\{(n,\tau_n): n \in \mathbb{Z}^+\right\}$ of the next-time variables. The remaining sets of indices are defined as
	\begin{align}
	\label{eq8b}
	\mathcal{B} = & \{(n+1,(2j+2)\tau_{n+1}): n,j \in \mathbb{Z}^+ \}, \\[3pt]
	\mathcal{R} = & \{(0,j+2): j \in \mathbb{Z}^+\} 
	\cup \{(n+1,(2j+3)\tau_{n+1}): n,j \in \mathbb{Z}^+\}.
	\label{eq8c}
	\end{align} 

\begin{proposition}
\label{theorem1}
For any given initial condition (\ref{eq7}), there is uncountable number of solutions of system (\ref{eq4}). Each solution is determined by initial conditions $u_n(t) \in \mathbb{T}^2$ for $(n,t) \in \mathcal{I}$ and arbitrary values $u_n(t) \in \mathbb{T}^2$ for $(n,t) \in \mathcal{R}$, in which case the remaining variables with $(n,t) \in \mathcal{G} \cup \mathcal{B}$ are defined uniquely.
\end{proposition}

\begin{proof}
Let us write equation (\ref{eq4}) as
	\begin{equation}
	\label{eq9}
	u_{n+1}(t) = A^{-1}u_n(t+\tau_n)-u_n(t)  \ \mathrm{mod}\ 1.
	\end{equation} 
Then, given arbitrary $N \in \mathbb{Z}^+$ and inspecting Fig.~\ref{fig1}, one can verify that all variables $u_n(t)$ with $n \le N$ and $(n,t) \in \mathcal{G}\cup\mathcal{B}$ are uniquely defined by the initial conditions at $(n,t) \in \mathcal{I}$ and the  variables with $n < N$ and $(n,t) \in \mathcal{R}$.
Hence, all equations (\ref{eq4}) and initial conditions (\ref{eq7}) are satisfied for arbitrary $u_n(t) \in \mathbb{T}^2$ at $(n,t) \in \mathcal{R}$ and uniquely defined variables at $(n,t) \in \mathcal{G}\cup\mathcal{B}$.
\end{proof}

We notice that solutions of Proposition~\ref{theorem1} include analogues of the so-called wild weak solutions for Euler equations in fluid dynamics~\cite{buckmaster2021convex}. These are unphysical solutions with a finite support in time, i.e., nonzero for $t \in (t_1,t_2)$ but vanishing both for $t \le t_1$ and $t \ge t_2$. Such solutions are constructed in our model by choosing the variables with $(n,t) \in \mathcal{R}$ to be zero for times $t \notin (t_1+1,t_2)$ and nonzero for $t \in (t_1+1,t_2)$, where $t_1$ and $t_2$ are arbitrary positive integers. One can show using Proposition~\ref{theorem1} and relation (\ref{eq9}) that this yields uncountably many wild solutions; see Fig.~\ref{fig_wild}.

\begin{figure}[t]
\centering
\includegraphics[width=0.6\textwidth]{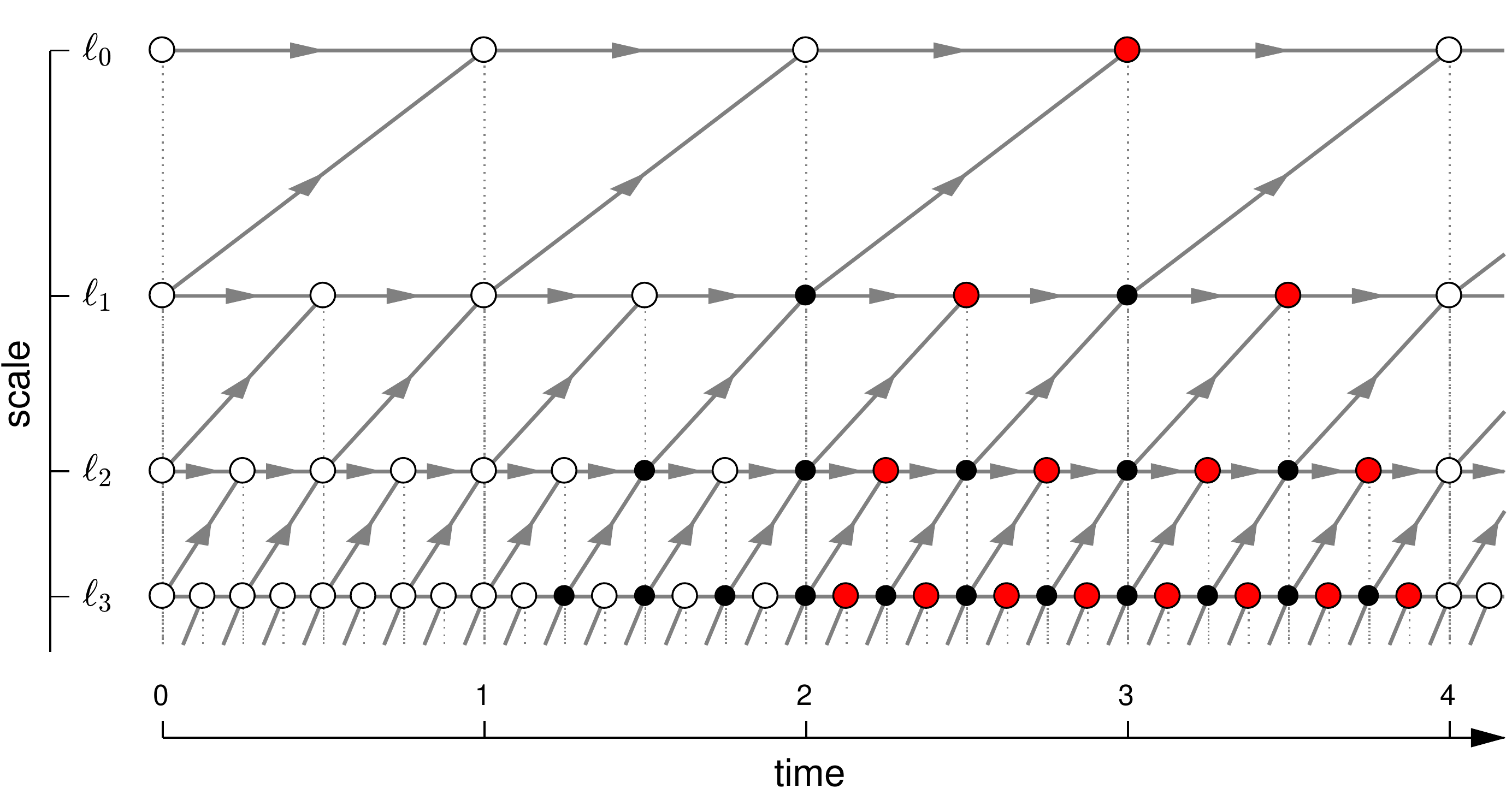}
\caption{Examples of wild solutions with the compact support in time: all variables vanish for $t \le 1$ and $t \ge 4$. Here all white circles correspond to zero variables and red circles denote arbitrary nonzero variables, which uniquely define the variables denoted by small black circles.}
\label{fig_wild}
\end{figure}

\section{Regularized solutions} \label{sec3}
Introducing a regularized system is a conventional way for dealing with non-uniqueness.
For any integer $N \ge 1$, we define the $N$-regularized system for the finite number of variables $u_1(t),\ldots,u_N(t)$ by setting the remaining variables to zero: $u_n(t) \equiv (0,0)$ for $n > N$. Thus, the set of variables reduces to $(u_n(t))_{(n,t) \in \mathcal{L}_N}$, where $\mathcal{L}_N = \left\{(n,t): n \in \{0,\ldots,N\},\ t \in \tau_n \mathbb{Z}^+\right\}$ is the truncated lattice. Equations of the $N$-regularized system are given by (\ref{eq4}) for $n < N$ with the equation for $n = N$ reduced to the form $u_n(t+\tau_n) = Au_n(t)$. The initial conditions are defined by relations (\ref{eq7}) limited to the scales $n \le N$. This truncation resembles the viscous regularization in fluid dynamics, where the viscous term of the Navier--Stokes equations suppresses the turbulent motion below a certain (so-called Kolmogorov) microscale $\eta \sim \ell_N$~\cite{frisch1999turbulence}. 

One can easily see from Fig.~\ref{fig1} that $N$-regularized solutions, which we denote by $u_n^{(N)}(t)$, are uniquely defined by the initial conditions. One can always choose a subsequence $N_1 < N_2 < \cdots$ such that the $N$-regularized solutions converge for all $n$ and $t$ (see~\cite[Theorem 3.10.35]{engelking1989general}):
	\begin{equation}
	\label{eq10}
	u_n(t) = \lim_{i \to \infty} u_n^{(N_i)}(t),
	\end{equation}
where $u_n(t)$ is some solution from Proposition~\ref{theorem1}. For example, for vanishing initial conditions, this limit yields the vanishing solution at all times $t \ge 0$, therefore, ruling out all wild solutions mentioned above.

Similarly to turbulence models~\cite{mailybaev2016spontaneous}, solutions obtained in the regularization limit (\ref{eq10}) are non-unique in general, because different subsequences yield different solutions: 

\begin{proposition}
\label{prop1}
Consider the initial condition (\ref{eq7}) with all variables equal to the same value $u_n^0 = a$. Then, for almost every choice of $a \in \mathbb{T}^2$, there exist infinite (uncountable) number of different solutions obtained as subsequence limits (\ref{eq10}) of  $N$-regularized systems.
\end{proposition}

\begin{proof}
Let us focus on the specific variable $u_0^{(N)}(2)$. By induction with relation (\ref{eq4}) represented by gray arrows in Fig.~\ref{fig1}, one can verify the formula
	\begin{equation}
	\label{eq11}
	u_0^{(N)}(2) = A^2 u_0^0 + (A+A^2)\sum_{n = 1}^{N} A^nu_n^0 \ \ \mathrm{mod}\ 1.
	\end{equation}
Taking into account that all initial values are equal to $a$, and $\sum_{n = 1}^N A^n = (A-I)^{-1}(A^{N+1}-A)$ with the identity map $I$, one reduces (\ref{eq11}) to the form
	\begin{equation}
	\label{eq12}
	u_0^{(N)}(2) = (A^2-B)a+BA^N a  \ \ \mathrm{mod}\ 1,
	\end{equation}
where $B = A^2(I+A)(A-I)^{-1}$ is a nonsingular matrix with integer components.

The ergodicity of the Arnold's Cat map implies that the sequence $A^Na$ with $N \in \mathbb{Z}^+$ is dense on the torus for almost every $a \in \mathbb{T}^2$. Let us consider such $a$, an arbitrary $c\in \mathbb{T}^2$ and define $b \in \mathbb{T}^2$ such that $(A^2-B)a+Bb=c$. Since $A^Na$ is a dense orbit, we can choose an infinite subsequence $N_i$ such that $A^{N_i}a \to b$ as $N_i \to \infty$. Then, expression (\ref{eq12}) yields 
	\begin{equation}
	\label{eq12bb}
	\lim_{i \to \infty} u_0^{(N_i)}(2)=c.
	\end{equation}
Similarly to (\ref{eq10}) we can take a subsequence $N_{i_k}$ within the sequence $N_i$, so that $u_n^{(N_{i_k})}(t)$ converges to a solution $u_n(t)$ for all $n$ and $t$. In particular, this implies that $u_0(2)=c$ for arbitrary $c\in \mathbb{T}^2$, providing uncountable number of limits for the regularized system.
\end{proof}

Proposition~\ref{prop1} shows that the regularization does not serve as a proper selection criterion among infinitely many solutions given by Proposition~\ref{theorem1}. As we show in the next section, there is a deep reason for this failure of the regularization strategy. Contrary to the common intuition, all solutions of Proposition~\ref{theorem1} become equally relevant when the stochastic form of regularization is considered.

\section{Spontaneously stochastic solution}\label{sec4}
Let us modify the definition of $N$-regularized solution by adding a random small-scale perturbation. For simplicity, we consider a single random number $\xi \in \mathbb{T}^2$ added to the initial value at the cuf-off scale $n = N$ as
	\begin{equation}
	\label{eq13}
	u_N^{(N)}(0) = u_N^0+\xi,
	\end{equation}
with $\xi$ having a Lebesgue integrable probability density $\rho(\xi)$. This formulation is not only technically convenient, but also highlights an exceptional role of even a single source of randomness at small scales. Generalization to multiple random sources is rather straightforward. 

Let us consider the mapping 
	\begin{equation}
	\label{eq13IC}
	\left(u_0^{0},\ldots,u_N^{0}\right) \mapsto 
	\left(u_n^{(N)}(t)\right)_{(n,t) \in \mathcal{L}}
	\end{equation}
relating deterministic initial conditions with deterministic $N$-regularized solutions; recall that $u_n^{(N)}(t) \equiv 0$ for $n > N$. For the new random initial condition (\ref{eq13}), we introduce the full vector of initial states as
	\begin{equation}
	\label{eq13ICa}
	(\zeta_0,\ldots,\zeta_N)  
	= \left(u_0^{0},\ldots,u_{N-1}^{0},u_N^{0}+\xi\right) \in \mathbb{T}^{2(N+1)},
	\end{equation}
and define the corresponding probability measure as
	\begin{equation}
	\label{eq13ICc}
	d\mu_{\mathrm{ini}}^{(N)} = \left(\prod_{n = 0}^{N-1} \delta(\zeta_n-u_N^{0})d\zeta_n\right) 
	\rho(\zeta_N-u_N^{0})d\zeta_N.
	\end{equation}
This measure is a product of Dirac delta functions on the torus $\mathbb{T}^2$ for the first $N$ components and the shifted density $\rho(\xi)$ for the last component. We denote by $\mu^{(N)}$ the corresponding probability measure of $N$-regularized solutions, which is naturally obtained as the image (push-forward) of $\mu_{\mathrm{ini}}^{(N)}$ by the mapping (\ref{eq13IC}).

Let us consider the standard product topology on the lattice $\mathcal{L}$ and Borel probability measures endowed with the weak-convergence topology; see e.g. \cite{tao2011introduction}. We say that the original problem (\ref{eq4})--(\ref{eq7}) has a \textit{spontaneously stochastic solution} described by a non-trivial measure $\mu$, if it is obtained as the limit
	\begin{equation}
	\label{eq_sslim}
	\mu = \lim_{N \to \infty} \mu^{(N)},
	\end{equation}
in which the regularization is removed. Now we can formulate our main result as

\begin{theorem}
\label{theorem2}
Problem (\ref{eq4})--(\ref{eq7}) has a spontaneously stochastic solution given by the probability measure $\mu$ specified below by Eqs.~(\ref{eq15a})--(\ref{eq14}). This measure is universal, i.e., independent of the small-scale perturbation $\xi$.
\end{theorem}

We postpone the proof for the next section, and now describe the measure $\mu$. This measure is composed as a product of four pieces. The first two are the probability measures $\mu_{\mathcal{I}}$ and $\mu_{\mathcal{G}}$ corresponding, respectively, to deterministic initial conditions (\ref{eq7}) and the next-time variables $u_n(\tau_n)$ given uniquely by relation (\ref{eq4}):
	\begin{align}
	\label{eq15a}
	d\mu_{\mathcal{I}} = & \prod_{n \in \mathbb{Z}^+} \delta\left(u_n(0)-u_n^0\right)du_n(0),
	\\[5pt]
	\label{eq15b}
	d\mu_{\mathcal{G}} = &\ \prod_{n \in \mathbb{Z}^+} 
	\delta\left(u_n(\tau_n)-Au_n^0-Au_{n+1}^0\right)\, du_n(\tau_n).
	\end{align}
Here $du_n(t)$ defines the Lebesgue (uniform) probability measure on $\mathbb{T}^2$ corresponding to a specific variable $u_n(t)$. The third piece is given by the measure 
	\begin{equation}
	\label{eq15c}
	d\mu_{\mathcal{R}} = \prod_{(n,t) \in \mathcal{R}} du_n(t),
	\end{equation}
which describes a random uniform choice of variables $u_n(t)$ from the red set $\mathcal{R}$; see Fig.~\ref{fig1}. The last piece ensures that all relations (\ref{eq4}) are satisfied for $t > \tau_n$. These relations are verified at points of the black set $\mathcal{B}$ using Eq.~(\ref{eq4}) transformed to form $u_n(t) = A^{-1}u_{n-1}(t+\tau_{n-1})-u_{n-1}(t)$; see Fig.~\ref{fig1}. Therefore, we define
	\begin{equation}
	\label{eq15d}
	d\mu_{\mathcal{B}} = \prod_{(n,t) \in \mathcal{B}}  
	\delta\left( u_n(t)+u_{n-1}(t)
	-\,A^{-1}u_{n-1}(t+\tau_{n-1}) \right) du_n(t).
	\end{equation}
The probability measure $\mu$ is given by the product
	\begin{equation}
	\label{eq14}
	d\mu = d\mu_{\mathcal{I}} \, d\mu_{\mathcal{G}} \, d\mu_{\mathcal{R}}\,  d\mu_{\mathcal{B}}.
	\end{equation} 

It is remarkable that the spontaneously stochastic solution $\mu$ assigns equal probability (uniform distribution) to all solutions of Proposition~\ref{theorem1} independently of the random perturbation $\xi$. One can see, however, that the probability measure corresponding to a set of wild solutions discussed above is zero.

Let us consider evolution of the spontaneously stochastic solution $\mu$ by focusing on integer times. At each $t \in \mathbb{Z}^+$, the solution defines a probability measure $\mu_t$ on the infinite-dimensional space of variables $\mathbf{u}(t) = \left(u_0(t),u_1(t),u_2(t),\ldots\right)$. For example, projecting the measure (\ref{eq15a})--(\ref{eq14}) at $t = 1$, we have 
	\begin{equation}
	\label{eq18b}
	d\mu_1 = \delta\left(u_0(1)-Au_0^0-Au_1^0\right)\, \prod_{n \in \mathbb{Z}^+}du_n(1),
	\end{equation}
where $u_0(1)$ is deterministic and $u_n(1)$ with $n \ge 1$ are random (independent and uniformly distributed). 
Measure (\ref{eq18b}) defines a Markov kernel: given a specific initial state $\mathbf{u}(0)$ it yields the probability distribution for $\mathbf{u}(1)$. Hence, the dynamics of our model at integer times represents a Markov process. In our example, $\mu_t$ converges at $t \ge 2$ to the equilibrium state $\mu_t \equiv \mu_{\mathrm{eq}}$, which is the uniform (Haar) measure on $\mathbb{T}^\infty$. As discussed in a different example later on, the convergence of $\mu_t$ as $t \to \infty$ does not always occur in a finite time.

By inspecting the proofs of Propositions~\ref{theorem1} and \ref{prop1} and of Theorem~\ref{theorem2} one can generalize our results as follows.

\begin{corollary} Let us consider a larger class of models given by relation (\ref{eq4}), where $A$ is an arbitrary $m \times m$ matrix with integer elements and $\det A =1$,
thus defining 
an automorphism of the $m$-dimensional torus $\mathbb{T}^m$. Proposition 1 remains valid with no additional hypothesis. Proposition \ref{prop1} is valid if we assume that $A$ does not possess eigenvalues which are roots of unity, in which case the induced automorphism of $\mathbb{T}^m$ is ergodic~\cite{Mane}. Finally, Theorem \ref{theorem2} remains valid under the two additional assumptions:
\begin{itemize}
\item[(i)] The dominant (maximum absolute value) eigenvalue $\lambda$ of $A$ is simple and greater than $1$.
\item[(ii)] Let $v = (v_1,\ldots,v_m)$ be the eigenvector corresponding to $\lambda$ for the transposed matrix $A^T$. Then the numbers $v_1,\ldots,v_m$ and $1$ are rationally independent. 
\end{itemize}
\end{corollary}

\section{Proof of Theorem~\ref{theorem2}}\label{sec5}
The weak convergence of measures $\mu^{(N)} \to \mu$ in the product topology~\cite{tao2011introduction} follows from the following property, which describes the convergence for all finite-dimensional projections.

\begin{lemma}
\label{theorem2b}
Let $\left(u_n(t)\right)_{(n,t) \in \mathcal{S}} \in \mathbb{T}^{2d}$ be any finite set of variables indexed by $\mathcal{S} = \{(n_i,t_i):i = 1,\ldots,d\} \subset \mathcal{L}$. Let $\mu_{\mathcal{S}}$ and $\mu_{\mathcal{S}}^{(N)}$ be the corresponding probability measures obtained by projecting the measure $\mu$ from (\ref{eq15a})--(\ref{eq14}) and the stochastically regularized measure $\mu^{(N)}$. Then, 
	\begin{equation}
	\label{eq18}
	\lim_{N \to \infty}\int \varphi \,d\mu^{(N)}_{\mathcal{S}} = \int \varphi \,d\mu_{\mathcal{S}}
	\end{equation}
for any continuous observable $\varphi:\mathbb{T}^{2d} \mapsto \mathbb{R}$.
\end{lemma}

\begin{proof} 
First, we express explicitly an arbitrary variable $u_n(t)$ in terms of initial conditions $u_n^0$ and the random quantity $\xi$ of the stochastically $N$-regularized problem. 
For this purpose, we use the polynomials $P_{n,t}^{(N)}$ defined as
	\begin{equation}
	\label{eq19A}
	P_{n,t}^{(N)} (x) = \sum_{\textrm{all paths }p\textrm{ from }\atop (N,0)\textrm{ to }(n,t)} x^{|p|},
	\end{equation}
where the sum is taken over all paths following grey (right or up-right diagonal) arrows in Fig.~\ref{fig1}, which connect $(N,0)$ to $(n,t)$, and $|p|$ denotes the number of arrows in the path. 
Using iteratively the linear relation (\ref{eq4}) with the truncation property $u_n(t) = 0$ for $n > N$ and the initial conditions (\ref{eq7}) and  (\ref{eq13}), one can check that
	\begin{equation}
	\label{eq20}
	u_n^{(N)}(t) = P_{n,t}^{(N)}(A) \xi+a_{n,t}^{(N)}\ \mathrm{mod}\ 1, \quad
	a_{n,t}^{(N)} = \sum_{k= n}^N P_{n,t}^{(k)}(A) u_{k}^0,
	\end{equation}
where $a_{n,t}^{(N)} \in \mathbb{T}^2$ denotes the contribution from deterministic initial conditions. The probability measure $\mu_{\mathcal{S}}^{(N)}$ corresponding to a finite set of random variables $\left(u_n^{(N)}(t)\right)_{(n,t) \in \mathcal{S}} \in \mathbb{T}^{2d}$ is obtained using relation (\ref{eq20}) as
	\begin{equation}
	\label{eq20b}
	d\mu_{\mathcal{S}}^{(N)}	= \int\left( \prod_{(n,t) \in \mathcal{S}} 
	\delta\left(u_n^{(N)}(t)-P_{n,t}^{(N)}(A) \xi-a_{n,t}^{(N)}\right)du_n^{(N)}(t)\right)
	\rho(\xi)\,d\xi,
	\end{equation}
where $du_n^{(N)}(t)$ denotes the Lebesgue (uniform) probability measure on $\mathbb{T}^2$ corresponding to a specific variable $u_n^{(N)}(t)$, and $\rho: \mathbb{T}^2 \mapsto \mathbb{R}^+$ is a measurable probability density for the random number $\xi$.

Now let us analyse an arbitrary set $\mathcal{S}$. It is enough to consider $\mathcal{S} = \mathcal{L}_{n,t}$ in the rectangular region
	\begin{equation}
	\label{eq24L}
	\mathcal{L}_ {n,t}= \left\{(n',t'): n' \le n,\ t' \le t\right\}
	\end{equation}
for any integer $n$ and $t$. Both measures $\mu$ and $\mu^{(N)}$ are supported on the linear subspace determined by relations (\ref{eq4}). Specifically, according to Proposition~\ref{theorem1} and Fig.~\ref{fig1}, variables at white nodes correspond to initial conditions, variables at green nodes are determined by initial conditions only, and variables at black nodes are given by initial conditions and by variables at red nodes of the set $\mathcal{R}$. These relations do not depend on $N$. Hence, for both projected measures $\mu_\mathcal{S}$ and $\mu_\mathcal{S}^{(N)}$ with $\mathcal{S} = \mathcal{L}_ {n,t}$, relations (\ref{eq4}) define variables $u_{n'}(t')$ from $\mathcal{S}$ in terms of initial conditions and variables from $\mathcal{L}_ {n,t} \cap \mathcal{R}$. From this property, one can infer that the relation (\ref{eq18}) can be verified for smaller sets of the form
	\begin{equation}
	\label{eq24S}
	\mathcal{S} = \mathcal{L}_ {n,t} \cap \mathcal{R},
	\end{equation}
in which we ignored the remaining deterministic variables.

Projecting the measure $\mu$ from (\ref{eq15a})--(\ref{eq14}) on the subspace given by (\ref{eq24S}), one obtains that $\mu_\mathcal{S}$ is the Lebesgue measure in $\mathbb{T}^{2d}$. Then, the integral in the right-hand side of (\ref{eq18}) reduces to the mean value of the observable:
	\begin{equation}
	\label{eq18mod}
	\lim_{N \to \infty}\int \varphi \,d\mu^{(N)}_{\mathcal{S}} = \int \varphi(\mathbf{w}) \, d^{2d} \mathbf{w},
	\end{equation}
where $\mathbf{w} = \left(w_{n,t}\right)_{(n,t) \in \mathcal{S}} \in \mathbb{T}^{2d}$ denotes the vector of variables indexed by $\mathcal{S}$, with $w_{n,t} = u_n^{(N)}(t)$ in the integral on the left-hand side.
Let is consider the Fourier expansion 
	\begin{equation}
	\label{eq25F}
	\varphi(\mathbf{w}) = \sum_{\mathbf{k} \in (2\pi\mathbb{Z})^{2d}} 
	\varphi_{\mathbf{k}} \exp(i\mathbf{k} \cdot \mathbf{w}),
	\end{equation}
where we introduced the wavevector $\mathbf{k} = \left(k_{n,t}\right)_{(n,t) \in \mathcal{S}} \in (2\pi\mathbb{Z})^{2d}$; the dot denotes the scalar product.
Using (\ref{eq25F}) in relation (\ref{eq18mod}), the constant term $\varphi_{\mathbf{0}}$ compensates the integral in the right-hand side since 
$\varphi_{\mathbf{0}} = \int \varphi \, d^{2d} \mathbf{w}$. 
Therefore, it remains to show that
	\begin{equation}
	\label{eq25}
	\lim_{N \to \infty} \int \exp(i\mathbf{k}\cdot \mathbf{w}) \, d\mu^{(N)}_{\mathcal{S}} = 0
	\end{equation}
for any nonzero wavevector $\mathbf{k}$.  Using (\ref{eq20b}) with the property $k_{n,t} \in (2\pi\mathbb{Z})^2$ and symmetry of the matrix $A$, we have
	\begin{equation}
	\label{eq26}
	\int \exp(i\mathbf{k}\cdot \mathbf{w}) \, d\mu^{(N)}_{\mathcal{S}} 
	= \exp\left(i a_{\mathbf{k}}^{(N)}\right) 
	\int \exp\left(i A_{\mathbf{k}}^{(N)} \cdot \xi\right)\rho(\xi)d\xi,
	\end{equation}
where we introduced the scalar $a_{\mathbf{k}}^{(N)} \in \mathbb{R}$ and the vector $A_{\mathbf{k}}^{(N)} \in \mathbb{R}^2$ as
	\begin{equation}
	\label{eq27}
	a_{\mathbf{k}}^{(N)} = \sum_{(n,t) \in \mathcal{S}} k_{n,t} \cdot a_{n,t}^{(N)},\quad
	A_{\mathbf{k}}^{(N)} = \sum_{(n,t) \in \mathcal{S}} P_{n,t}^{(N)}(A)k_{n,t}.
	\end{equation}
Notice that the integral in the right-hand side of (\ref{eq26}) represents the Fourier coefficient of $\rho(\xi)$ of order $-A_{\mathbf{k}}^{(N)}$. By the Riemann--Lebesgue lemma the high-order Fourier coefficients of the function $\rho(\xi)$ converge to zero. Therefore, to conclude the proof it is enough to show that $\|A_{\mathbf{k}}^{(N)}\| \rightarrow \infty$ as $N\rightarrow\infty$ for any fixed nonzero wavevector $\mathbf{k} \in (2\pi\mathbb{Z})^{2d}$. 
	
Using the eigenvalue decomposition of the Arnold's cat map (\ref{eq6}), we can write~\cite{arnold1968ergodic} 	\begin{equation}
	\label{eq28}
	P_{n,t}^{(N)}(A) = P_{n,t}^{(N)}(\alpha) A_1+P_{n,t}^{(N)}(\alpha^{-1}) A_2,
	\end{equation}
where $\alpha = (3+\sqrt{5})/2$ and $\alpha^{-1}$ are eigenvalues of $A$, and the symmetric matrices $A_1$ and $A_2$ are given by the linear maps
	\begin{equation}
	\label{eq28A}
	\textstyle
	A_1:(x,y) \mapsto \left(\frac{\alpha x+(\alpha-1)y}{\alpha+1},\frac{(\alpha-1)x+y}{\alpha+1}\right), \quad
	A_2:(x,y) \mapsto \left(\frac{x+(1-\alpha)y}{\alpha+1},\frac{(1-\alpha)x+\alpha y}{\alpha+1}\right).
	\end{equation}
Substituting (\ref{eq28}) into the second expression of (\ref{eq27}) yields
	\begin{equation}
	\label{eq27_new}
	A_{\mathbf{k}}^{(N)} = \sum_{(n,t) \in \mathcal{S}} 
	\left[P_{n,t}^{(N)}(\alpha) A_1k_{n,t}
	+P_{n,t}^{(N)}(\alpha^{-1}) A_2k_{n,t}\right].
	\end{equation}
Since $P_{n,t}^{(N)}$ defined in (\ref{eq19A}) is a polynomial with positive coefficients and $\alpha > 1$, we have 
	\begin{equation}
	\label{eq29}
	\lim_{N \to \infty} \frac{P_{n,t}^{(N)}(\alpha)}{P_{n,t}^{(N)}(\alpha^{-1})} = \infty.
	\end{equation}
Using Lemma~\ref{pol_conv} formulated and proved below, we can order the elements in $\mathcal{S} = \{(n_i,t_i): i =1,\ldots,d\}$ such that 
	\begin{equation}
	\label{eq30}
	\lim_{N \to \infty} \frac{P_{n_{i+1},t_{i+1}}^{(N)}(\alpha)}{P_{n_i,t_i}^{(N)}(\alpha)} 
	= \infty, \quad i = 1,\ldots,d-1.
	\end{equation}
	
Notice that $A_1k_{n,t} = \frac{1}{\alpha+1}\left(\alpha \ \ \alpha-1 \atop \alpha-1 \ \ 1 \right) k_{n,t}$ following from (\ref{eq28A}), where $\alpha = (3+\sqrt{5})/2$ is the irrational number. Since the wavevector $k_{n,t}/(2\pi) \in \mathbb{Z}^2$ has integer components, $A_1k_{n,t}$ is nonzero if $k_{n,t}$ is nonzero. Therefore, using
properties (\ref{eq29}) and (\ref{eq30}) in expression (\ref{eq27_new}), one can see that the magnitude of $A_{\mathbf{k}}^{(N)}$ is dominated by the polynomial $P_{n_i,t_i}^{(N)}(\alpha)$ with the largest $i$ such that $k_{n_i,t_i}$ is nonzero. Since $P_{n_i,t_i}^{(N)}(\alpha) \to \infty$ as $N \to \infty$, we prove the desired property that $\|A_{\mathbf{k}}^{(N)}\| \rightarrow \infty$ as $N\rightarrow\infty$. 
\end{proof}

\begin{lemma} \label{pol_conv}
Elements $(n_j, t_j)$, $j = 1,\ldots,d$ of any finite subset $\mathcal{S} \subset \mathcal{R}$ can be ordered such that  (\ref{eq30}) holds.
\end{lemma}

\begin{figure}
\centering
\includegraphics[width=0.65\textwidth]{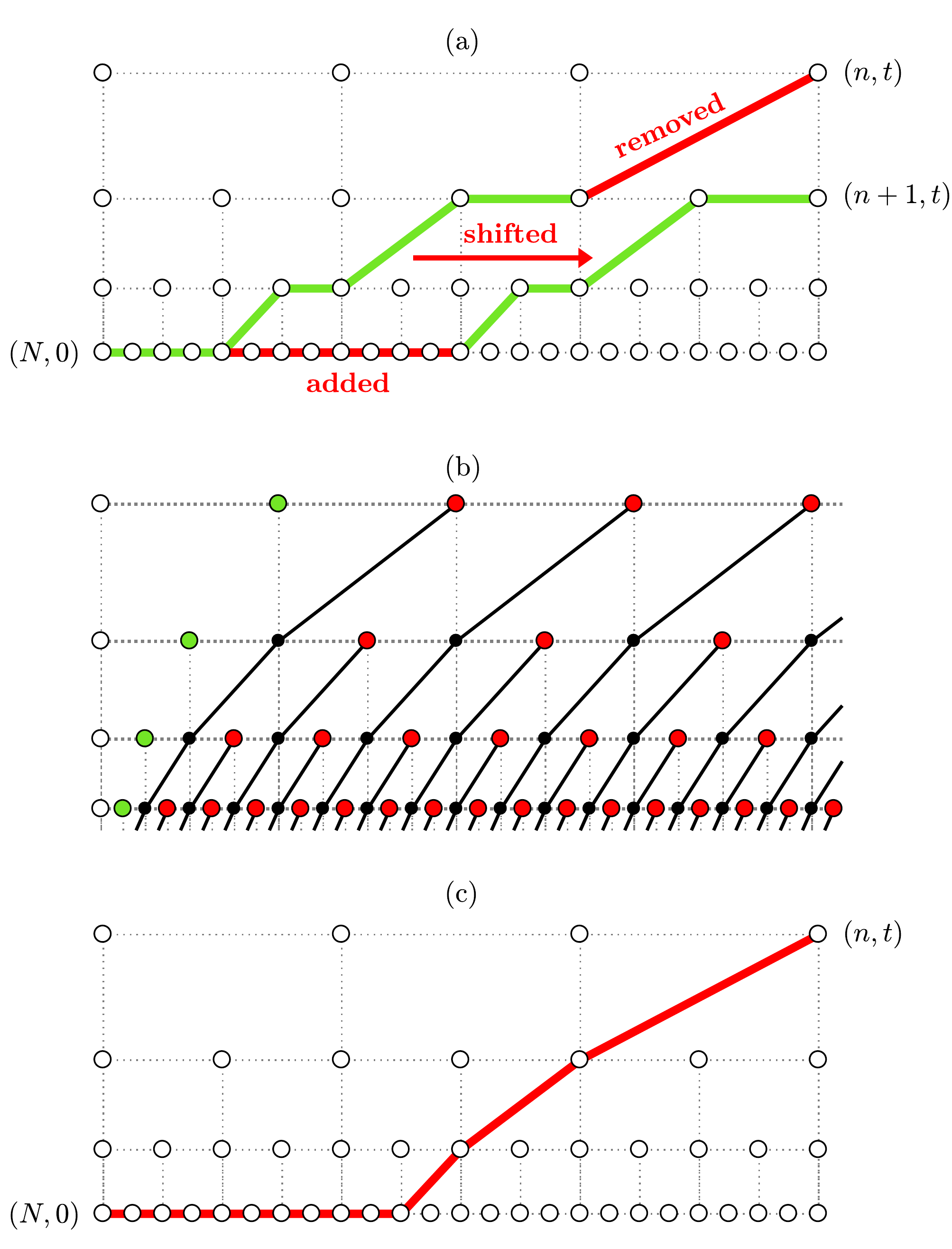}
\caption{(a) Every path connecting $(N,0)$ to $(n,t)$ defines the path connecting $(N,0)$ to $(n+1,t)$ through the following surgery procedure. The upper (red) part of the path is removed and the remaining (green) part is shifted to the left. Then, the lower (red) part is added to complete the new path. (b) Black lines connect nodes $(n',t')$ related by Eq.~(\ref{eq23b}), where $(n,t)$ are taken at red points. Polynomials on the same line have finite (nonzero and non-infinite) ratios in the limit $N \to \infty$. (c) The path connecting $(N,0)$ to $(n,t)$ with the largest number of segments.}
\label{figA}
\end{figure}

\begin{proof}
Observe that the condition $(n,t) \in \mathcal{R}$ with $n \le N$ ensures that $P_{n,t}^{(N)} (x)$ from (\ref{eq19A}) is nonzero for any $x > 0$; see Fig.~\ref{fig1}. For any path $p$ from $(N,0)$ to $(n,t)$ in (\ref{eq19A}), one constructs a new path $p'$ from $(N,0)$ to $(n+1,t)$ as shown in Fig.~\ref{figA}(a): removing the final segments at scale $n$, shifting the remaining part to the right, and adding extra segments at scale $N$. In this procedure, each removed segment yields the $\tau_n/\tau_N = 2^{N-n}$ added segments. This means that
	\begin{equation}
	\label{eq21}
	\lim_{N \to \infty} \frac{P_{n+1,t}^{(N)}(x)}{P_{n,t}^{(N)}(x)} = \infty,
	\end{equation}
where we assumed an arbitrarily chosen number $x > 1$. 
Notice that the definition (\ref{eq19A}) implies
	\begin{equation}
	\label{eq22}
	P_{n,t}^{(N)} (x) = xP_{n,t-\tau_n}^{(N)} (x)+xP_{n+1,t-\tau_n}^{(N)} (x),
	\end{equation}
where the last two terms correspond to the paths ending, respectively, with the horizontal and diagonal arrows (Fig.~\ref{fig1}).
Using (\ref{eq21}) in (\ref{eq22}), we have
	\begin{equation}
	\label{eq23}
	\lim_{N \to \infty} \frac{P_{n,t}^{(N)}(x)}{P_{n+1,t-\tau_n}^{(N)}(x)} = x.
	\end{equation}
Iterating this relation yields
	\begin{equation}
	\label{eq23b}
	\lim_{N \to \infty} \frac{P_{n,t}^{(N)}(x)}{P_{n',t'}^{(N)}(x)} = x^{n'-n}
	\quad \textrm{for}\quad n' > n,\quad t' = t-\sum_{j = n}^{n'-1}\tau_j.
	\end{equation}
When $(n,t) \in \mathcal{R}$, the points $(n',t')$ from (\ref{eq23b}) belong to a descending diagonal line as shown in Fig.~\ref{figA}(b). Inspecting these diagonal lines and using the property (\ref{eq21}), one can deduce that			
	\begin{equation}
	\label{eq24extr}
	\lim_{N \to \infty} \frac{P_{n_2,t_2}^{(N)}(x)}{P_{n_1,t_1}^{(N)}(x)} = \infty
	\end{equation}
for any distinct elements $(n_1,t_1)$ and $(n_2,t_2)$ of the set $\mathcal{R}$. Here the indices are chosen such that the black line starting at $(n_2,t_2)$ is located to the right of the line starting at $(n_1,t_1)$; see Fig.~\ref{figA}(b). In particular, this implies that any finite subset of elements $(n_j,t_j) \in \mathcal{R}$ can be ordered satisfying the properties (\ref{eq30}).
\end{proof}

\section{Convergence rate}\label{sec6}
We now address practical aspects of convergence: how small can be the random perturbation $\xi$ and how large must be the number of scales $N$ for observing the spontaneously stochastic solution with a given variable $u_n(t)$? 
	
Relations (\ref{eq20}) and (\ref{eq28}) in the proof of Theorem~\ref{theorem2} with the limit (\ref{eq29}) indicate that the convergence to the spontaneously stochastic limit for each variable $u_n^{(N)}(t)$ is controlled by the factor 
	\begin{equation}
	\label{eq19}
	P_{n,t}^{(N)} (\alpha) = \sum_{\textrm{all paths }p\textrm{ from }\atop (N,0)\textrm{ to }(n,t)} \alpha^{|p|}.
	\end{equation}
Here $\alpha = \frac{1}{2}(3+\sqrt{5}) \approx 2.618$ and the sum is taken over all paths following grey (right or up-right diagonal) arrows in Fig.~\ref{fig1}, which connect $(N,0)$ to $(n,t)$; $|p|$ denotes the number of arrows the path. The factor (\ref{eq19}) amplifies the random perturbation induced by $\xi$ in the variable $u_n^{(N)}(t)$. Let us assume that $\xi$ takes small random values of order $\varepsilon$ and has a sufficiently regular probability density (e.g. Holder continuous). Hence, for observing spontaneous stochasticity at node $(n,t)$, the corresponding error must become large: $P_{n,t}^{(N)}(\alpha) \varepsilon \gg 1$. This yields the condition
	\begin{equation}
	\label{eq31}
	P_{n,t}^{(N)}(\alpha) \gg 1/\varepsilon.
	\end{equation}

\begin{figure}[t]
\centering
\includegraphics[width=0.55\textwidth]{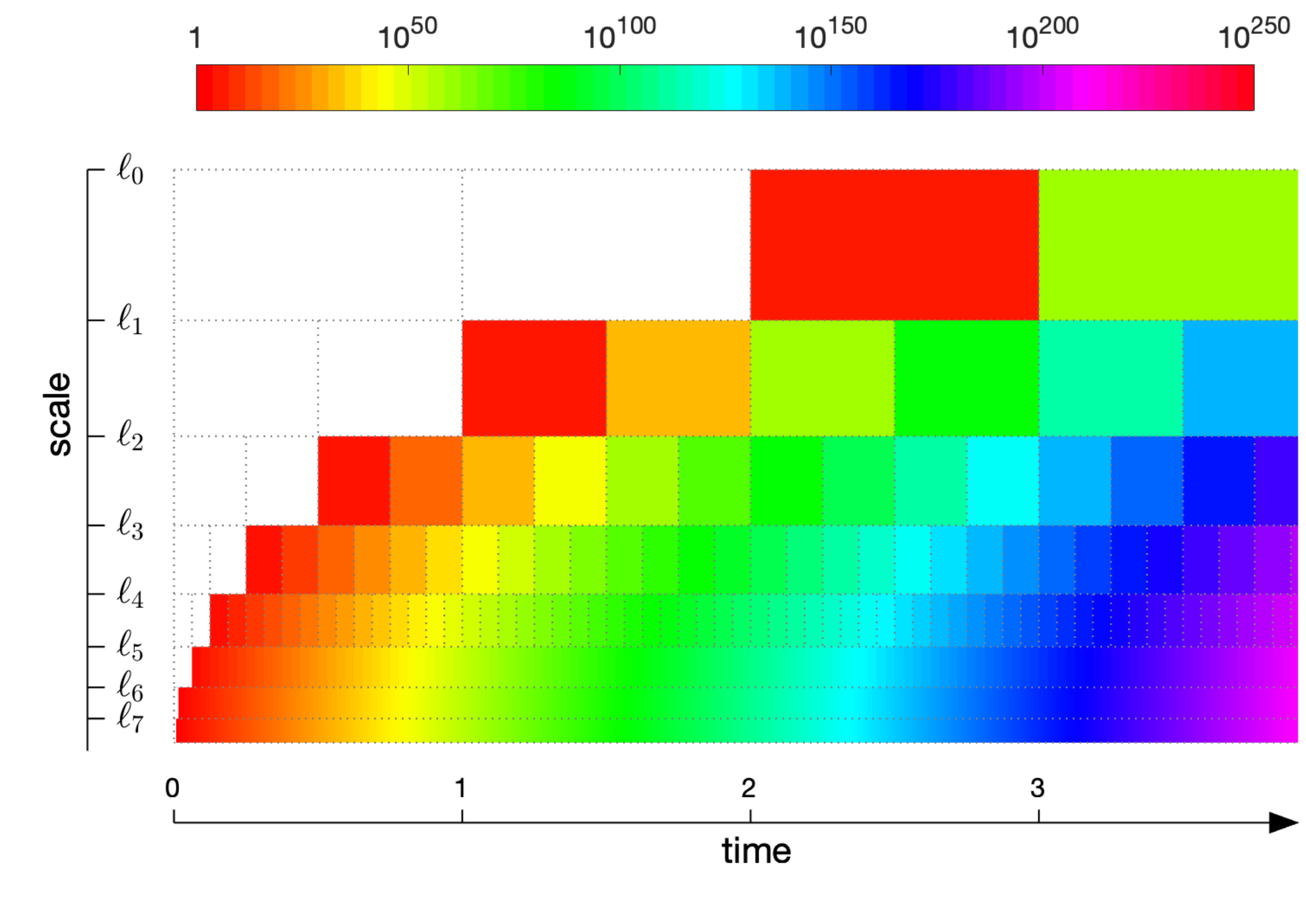}
\caption{Amplification factor $P_{n,t}^{(N)}(\alpha)$ of the initial error evaluated at each point of the lattice in a system with $N = 7$ scales. The color of each rectangle shows (in logarithmic scale) the value of $P_{n,t}^{(N)}(\alpha)$ corresponding to the node $(n,t)$ located in the upper left corner of the rectangle; zero values are shown by white color.}
\label{fig2}
\end{figure}

\begin{figure}
\centering
\includegraphics[width=0.8\textwidth]{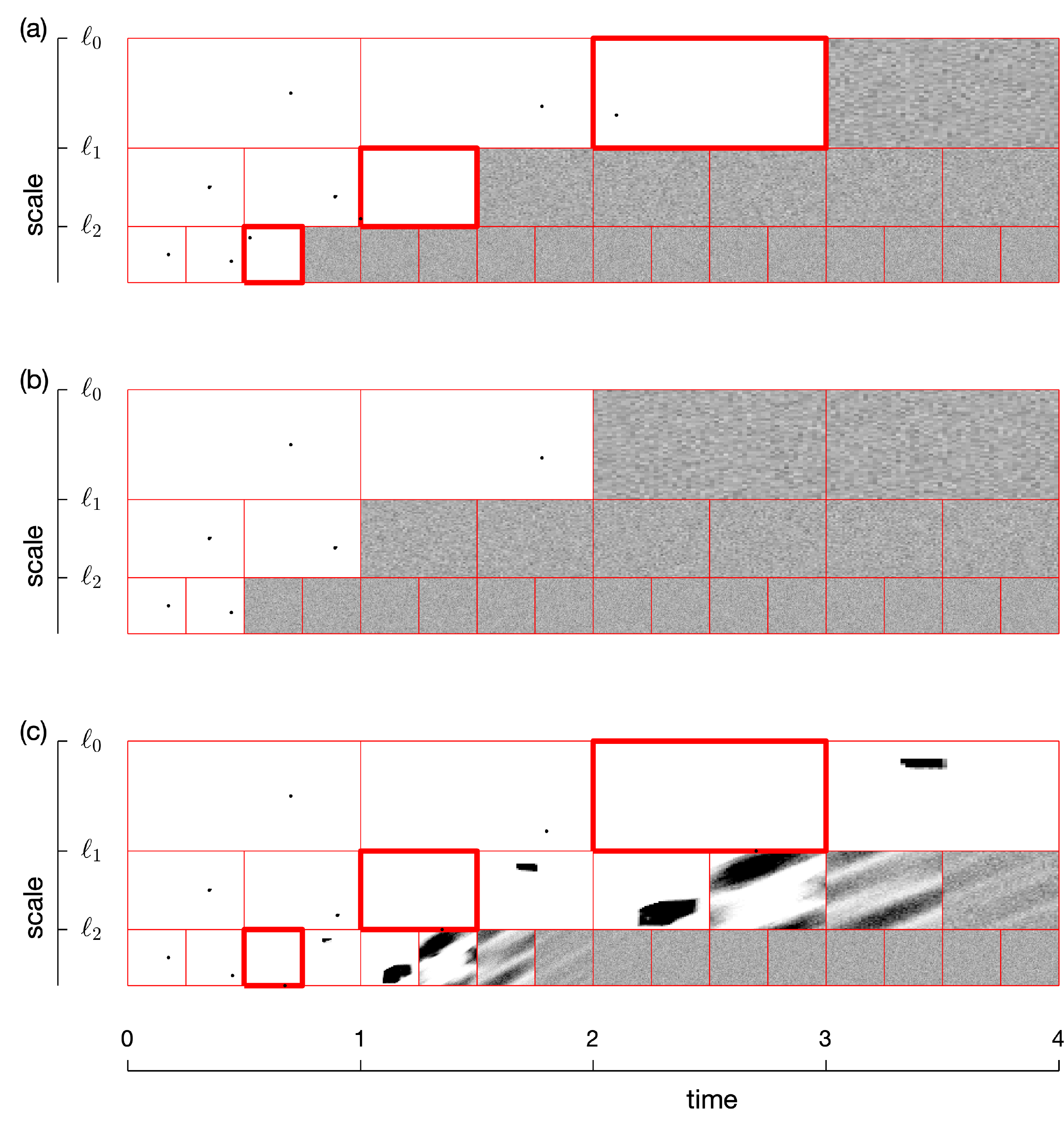}
\caption{Each rectangle shows the probability density functions (darker colors correspond to higher probabilities) for the variable $u_n(t) \in \mathbb{T}^2$ corresponding to the upper left corner of the rectangle. Only the scales $n = 0,1,2$ are demonstrated. The results are obtained by simulating numerically $10^8$ samples of the system with the initial conditions $u_n(0) = (0.7,0.5)$ for $n = 0,\ldots,N$ and the random variable $\xi$ uniformly distributed in the interval $[0,\varepsilon]$: (a) $N = 7$ and $\varepsilon = 10^{-10}$, and (b) $N = 9$ and $\varepsilon = 10^{-1}$. Bold red borders designate variables in the transitional region, $t \sim 2\tau_n$, where the convergence is only exponential in $N$ and is not attained for very small $\varepsilon$. The last panel (c) shows analogous results for $N = 9$ and $\varepsilon = 10^{-10}$ in modified system (\ref{eq4new}). }
\label{fig3}
\end{figure}

We now verify how fast $P_{n,t}^{(N)}(\alpha)$ grows with $N$. For this purpose, we compute the longest path (dominant term) in expression (\ref{eq19}).
This path contains the maximum number of arrows at the smallest scale $\ell_N$, supplemented with $N-n$ diagonal arrows (one at every scale) in order to reach the node $(n,t)$; see Fig.~\ref{figA}(c). The number of arrows at scale $N$ is evaluated as $(t-\Delta t)/\tau_N$, where $\tau_N = 2^{-N}$ is the turn-over time for each arrow and $\Delta t$ is the time interval occupied by the diagonal arrows at larger scales. This interval is evaluated as
	\begin{equation}
	\label{eqS1}
	\Delta t = \sum_{j = n}^{N-1}\tau_j 
	= \sum_{j = n}^{N-1} 2^{-j} 
	= 2^{1-n}-2^{1-N} = 2\tau_n-2\tau_N.
	\end{equation}
Therefore, the total number of arrows in the path is found as
	\begin{equation}
	\label{eqS2}
	|p| = \frac{t-\Delta t}{\tau_N}+N-n
	= \frac{t-2\tau_n}{\tau_N}+2+N-n
	= 2^N(t-2\tau_n)+N-n+2.
	\end{equation}

Using the longest path (\ref{eqS2}) in expression (\ref{eq19}), yields the lower-bound estimate as
	\begin{equation}
	\label{eq32}
	P_{n,t}^{(N)}(\alpha) \ge \alpha^b,\quad 
	b = 2^N(t-2\tau_n)+N-n+2.
	\end{equation}
This expression suggests that the time $t = 2\tau_n$ is transitional: the factor $P_{n,t}^{(N)}(\alpha) \propto \alpha^N$ grows exponentially in $N$ at $t \sim 2\tau_n$, while the growth becomes double-exponential with $P_{n,t}^{(N)}(\alpha) \propto \left(\alpha^{t-2\tau_n}\right)^{2^N}$ at larger times.

To be more specific, we computed the values of $P_{n,t}^{(N)}(\alpha)$ numerically using formula (\ref{eq19}) and presented the results graphically in Fig.~\ref{fig2}. One observes that, in the model with only $N = 7$ scales and utterly small noise of amplitude $\varepsilon \sim 10^{-50}$, the spontaneously stochastic behaviour develops for all variables lying to the right of the transitional (red/yellow) region. Therefore, systems with a moderate number of scales $N$ must demonstrate the spontaneously stochastic behaviour even for extremely small random perturbations. However, larger perturbations are required for convergence in the transitional region. This result is tested numerically in Figs.~\ref{fig3}(a,b).

\section{Discussion}\label{sec7}
We designed a simple model that demonstrates the Eulerian spontaneous stochasticity (ESS): It is a formally deterministic scale-invariant system with deterministic initial conditions, which has uncountably many non-unique solutions and yields a universal stochastic process when regularized with a small-scale infinitesimal random perturbation. Our work provides the rigorous study of this system proving the existence of spontaneously stochastic solution as well as its universality (independence of the vanishing regularization term). The exceptional and counterintuitive property of this solution is that it assigns equal probability (uniform probability density) to all non-unique solutions. At integer times, the solution represents a Markov process converging to the equilibrium (uniform) state. 

Our results can be extended to other forms of random regularization, e.g., random variables depending on $N$ or random perturbations added to all variables (noise). Also, one can use this idea for designing spontaneously stochastic systems with different behaviors by modifying the couplings or imposing extra conditions like conserved quantities. For example, Fig.~\ref{fig3}(c) shows the numerical results when Eq.~(\ref{eq4}) is replaced by
	\begin{equation}
	\label{eq4new}
	u_n(t+\tau_n) = Au_n(t)+0.1\left(\cos x_{n+1}(t),\cos y_{n+1}(t)\right),
	\end{equation}
where $u_{n+1}(t) = \left(x_{n+1}(t),y_{n+1}(t)\right) \in \mathbb{T}^2$. We see that model (\ref{eq4new}) yields a more sophisticated spontaneously stochastic solution. The rigorous study of such systems is challenging, leaving important theoretical questions for future study: how to analyse the existence, universality and robustness of spontaneously stochastic solutions in more general multi-scale models?

Our model can used as a prototype for a (first) experimental observation of the ESS implemented in a physical system, e.g., an optical or electric circuit. In this experiment, arrows in Fig.~\ref{fig1} represent waveguides, and coupling nodes are identical signal-processing gates. The scaling symmetry is maintained by choosing lengths of connecting waveguides proportional to turn-over times $\tau_n$, exploiting the property that a distance travelled by a signal is proportional to time.
The variables $u_n(t)$ can describe phases of propagating signals measured at each node, while the initial conditions are associated with the input signal. A challenge of this setup is in reproducing the coupling relation (\ref{eq4}) or a similar one that leads to the spontaneous stochasticity. Notice that the intrinsic hyperbolicity of Arnold's cat map can also be recreated in a simple mechanical system~\cite{hunt2003anosov,kuznetsov2005example}. The extremely fast convergence, which is double-exponential in a number of scales, suggests that the spontaneous stochasticity in the described experiment will be triggered by a natural microscopic noise  from the environment already in systems of moderate size, e.g. $N = 7$ from Fig.~\ref{fig2}.

Finally, the proposed model suggests that applications and occurrence of the ESS can be seen in a broader sense. This refers to multi-scale systems defined by deterministic rules but generating complex and genuinely stochastic processes. In real-world systems, this stochasticity may be triggered by a natural microscopic noise. Confirming ESS experimentally would imply that the occurrence of ESS should be studied in a wide range of applications, e.g., hydrodynamic turbulence, random-number generation, neural networks in artificial intelligence or living organisms, etc.
 
\vspace{2mm}\noindent\textbf{Acknowledgments.} 
The work was supported by CNPq (grants 303047/2018-6, 406431/2018-3).

\bibliographystyle{plain}
\bibliography{refs}

\end{document}